\documentclass[10pt,reqno]{amsart}
     \makeatletter
     \def\section{\@startsection{section}{1}%
     \z@{.7\linespacing\@plus\linespacing}{.5\linespacing}%
     {\bfseries
     \centering
     }}
     \def\@secnumfont{\bfseries}
     \makeatother
\newtheorem{theorem}{Theorem}[section]

\newtheorem{proposition}[theorem]{Proposition}

\theoremstyle{definition}

\theoremstyle{remark}
\newtheorem{remark}[theorem]{Remark}
\numberwithin{equation}{section}
\usepackage{graphicx}
\usepackage{float}
\begin{document}

\title[Modelling Illiquid Stocks: Asymptotic Methods]{Modelling Illiquid Stocks Using Quantum Stochastic Calculus: Asymptotic Methods}

\author{Will Hicks}
\address{Will Hicks: Memorial University of Newfoundland, St. John's, NL A1C 5S7, Canada}
\email{williamh@mun.ca}



\begin{abstract}
This article investigates the Fokker-Planck equations that arise from the application of quantum stochastic calculus to the modelling of illiquid financial markets, using asymptotic methods. We present a power series solution for quantum stochastic processes with a non-zero conservation process. Whilst the series in question are in general divergent, we show they can be used to approximate solutions for longer time frames, and provide estimates for the relative error on the higher order terms.
\end{abstract}

\maketitle


\section{Introduction}

The analysis in \cite{AB} shows how to apply the methods of quantum stochastic calculus developed in \cite{HP}, to derive a general form for a Quantum Black Scholes equation.

The article \cite{Hicks6} provides an example of where the underlying quantum stochastic process incorporates a non-zero conservation process in addition to the creation \& annihilation processes. The resulting random motion of the underlying traded asset price shows non-Gaussian moments, and the associated Fokker-Planck equation is a linear partial differential equation with an infinite number of terms (see also \cite{Hicks3}).

In this article we investigate ways in which we can generate asymptotic solutions to the models developed in \cite{Hicks6}. The resulting solutions, in connection with the discussion in \cite{Hicks6}, can be used in the study of the dynamics of illiquid stocks with a non-zero bid-offer spread.

We start in section \ref{Theory} by giving an overview of the theoretical background, before deriving the asymptotic solution in section \ref{Sln}. In section \ref{Conv}, we prove key results regarding the convergence of the solution, and finally investigate some numerical examples in section \ref{NumSln}.
\section{Theoretical Background}\label{Theory}
In this section, we summarise the analysis presented in \cite{Hicks6} in order to provide the necessary context for the modelling problem that we address using the asymptotic series in section \ref{Sln}.

This specific problem provides an example where modelling using a non-zero conservation process, and by extension the asymptotic methods presented in this article, can be useful.

Note, this section is intended as an overview of the background regarding the final partial differential equation:
\begin{align}\label{the pde}
\frac{\partial p}{\partial t}&=\sigma^2\sum_{k\geq 1}\frac{\epsilon^{(2k-2)}}{(2k)!}\frac{\partial^{2k} p}{\partial x^{2k}}+\sigma^2\eta\sum_{k\geq 2}\frac{(-\epsilon)^{(2k-3)}}{(2k-1)!}\frac{\partial^{(2k-1)} p}{\partial x^{(2k-1)}}
\end{align}
Readers interested only in the asymptotic methods used to derive a solution to equation \ref{the pde}, can skip to section \ref{Sln}. Alternatively, for more detail see \cite{Hicks6}.
\subsection{Hilbert Space Representation of the Financial Market:}
Many models of the financial market consider a single market price for each tradable asset as the random variable of interest. Furthermore, if one wishes to apply the methods of quantum probability, one could consider an observable $X$, acting on $\mathcal{H}\in L^2(\mathbb{R})$:
\begin{align*}
(X\psi)(x)&=x\psi(x)\text{, for }\psi(x)\in L^2(\mathbb{R})
\end{align*}
In this article, we consider instead a market made up of a number of buyers who would like to buy at the lower {\em bid} price, and sellers who would like to sell at the higher {\em offer} price.

Therefore, we consider the case where there are 2 state variables. One coordinate: $x$, that represents the mid-price for the traded asset, and a second coordinate: $\epsilon$ that represents the width of the bid-offer spread.

We assume that the state of the market for potential buyers (and sellers) is determined by wave functions in the Hilbert space of complex valued square integrable functions on $\mathbb{R}^2$:
\begin{align}\label{psi_o,b}
\psi_o(x,\epsilon)\in L^2(\mathbb{R}^2,\mathbb{C})\\
\psi_b(x,\epsilon)\in L^2(\mathbb{R}^2,\mathbb{C})\nonumber
\end{align}
The overall market state is defined by the direct sum:
\begin{align}\label{Market_def}
\psi &=\psi_o\oplus\psi_b\\
\psi_o(x,\epsilon)\text{, }\psi_b(x,\epsilon) &\in L^2(\mathbb{R}^2)\nonumber
\end{align}
For $\phi=\phi_1\oplus\phi_2$ and $\psi=\psi_1\oplus\psi_2$, we have:
\begin{align*}
\langle\phi|\psi\rangle &=\langle\phi_1|\psi_1\rangle+\langle\phi_2|\psi_2\rangle
\end{align*}
So it follows that the normalisation condition becomes:
\begin{align}\label{norm_cond}
||\psi_0\oplus\psi_b||^2&=||\psi_o||^2+||\psi_b||^2\nonumber\\
&=1
\end{align}
For example, we may have an even balance of buyers \& sellers, in which case:
\begin{align*}
||\psi_o||^2=||\psi_b||^2=1/2
\end{align*}
In general, as long as the normalization condition, given by equation \ref{norm_cond}, is met then we can have:
\begin{itemize}
\item More buyers than sellers: $||\psi_b||^2>||\psi_o||^2$.
\item More sellers than buyers: $||\psi_o||^2>||\psi_b||^2$.
\end{itemize}
\begin{remark}
Going forward, we make use of matrix notation, so that for $\psi\in \mathcal{S}(\mathbb{R})\oplus\mathcal{S}(\mathbb{R})$ we write:
\begin{align*}
|\psi\rangle &=\begin{pmatrix}\psi_0\\ \psi_b\end{pmatrix}\\
A\psi&=\begin{pmatrix} A_{11}& A_{12}\\A_{21}&A_{22}\end{pmatrix}\begin{pmatrix}\psi_0\\ \psi_b\end{pmatrix}
\end{align*}
Note that, we also apply the following abuse of notation, by writing:
\begin{align*}
\langle\psi|&=\begin{pmatrix}\psi_o&\psi_b\end{pmatrix}
\end{align*}
So that we write:
\begin{align*}
E^{\psi}[A]&=\langle\psi|A|\psi\rangle\\
&=\begin{pmatrix}\psi_o&\psi_b\end{pmatrix}\begin{pmatrix} A_{11}& A_{12}\\A_{21}&A_{22}\end{pmatrix}\begin{pmatrix}\psi_0\\ \psi_b\end{pmatrix}\\
&=\langle\psi_o|A_{11}|\psi_o\rangle+\langle\psi_o|A_{12}|\psi_b\rangle+\langle\psi_b|(A_{21}|\psi_o\rangle+\langle\psi_b|A_{22}|\psi_b\rangle
\end{align*}
\end{remark}
We define the price operator:
\begin{align}\label{X_com}
X=\begin{pmatrix} x+\epsilon/2&0\\0&x-\epsilon/2\end{pmatrix}
\end{align}
so that if we have:
\begin{align}\label{sellers}
|\psi\rangle&=\begin{pmatrix}\psi_o(x,\epsilon)\\0\end{pmatrix}
\end{align}
We get:
\begin{align*}
E^{\psi}[X]&=\begin{pmatrix}\psi_o(x,\epsilon)&0\end{pmatrix}\begin{pmatrix} x+\epsilon/2&0\\0&x-\epsilon/2\end{pmatrix}\begin{pmatrix}\psi_0(x,\epsilon)\\0\end{pmatrix}\\
&=\int_{\mathbb{R}^2}(x+\epsilon/2)||\psi_o(x,\epsilon)||^2dxd\epsilon\\
&=x_o
\end{align*}
Similarly, if we have:
\begin{align}\label{buyers}
|\psi\rangle&=\begin{pmatrix}0\\\psi_b(x,\epsilon)\end{pmatrix}
\end{align}
We get:
\begin{align*}
E^{\psi}[X]&=\begin{pmatrix}0&\psi_b(x,\epsilon)\end{pmatrix}\begin{pmatrix} x+\epsilon/2&0\\0&x-\epsilon/2\end{pmatrix}\begin{pmatrix}0\\ \psi_b(x,\epsilon)\end{pmatrix}\\
&=\int_{\mathbb{R}^2}(x-\epsilon/2)||\psi_o(x,\epsilon)||^2dxd\epsilon\\
&=x_b
\end{align*}
\subsection{Introducing a Quantum Stochastic Process:}
We introduce randomness to \ref{X_com} using the approach outlined in \cite{AB} (see also \cite{HP}), we take the tensor product of $\mathcal{H}$ with the symmetric Fock space: $\mathcal{H}\otimes\Gamma(L^2(\mathbb{R}^+;\mathbb{C}))$, and use a unitary time evolution operator to build the price operator at $t=T$.

If the price operator at $t=0$ is written: $X\otimes\mathbb{I}$, then the operator at $t=T$ is given by: $j_T(X)=U_T^*(X\otimes\mathbb{I})U_T$. $U_t$ is defined by the process (see \cite{HP} proposition 7.1):
\begin{align}\label{U_QSP}
dU_t=-\bigg(\Big(iH+\frac{L^*L}{2}\Big)\otimes dt+L^*S\otimes dA_t-L\otimes dA^{\dagger}_t+(\mathbb{I}-S)\otimes d\Lambda_t\bigg)U_t
\end{align}
Whereby $H,L$, and $S$ act on $\mathcal{H}$, and $dA_t,dA^{\dagger}_t$, and $d\Lambda_t$ act on the Fock space. By writing out (see \cite{HP} Theorem 4.5):\begin{align*}
dj_t(X) &=d(U_t^*(X\otimes\mathbb{I})U_t)\\
&=dU_t^*(X\otimes\mathbb{I})U_t+U_t^*(X\otimes\mathbb{I})dU_t+dU_t^*(X\otimes\mathbb{I})dU_t
\end{align*}
and using It{\^ o} multiplication: Table \ref{ito_table} (see \cite{HP}), we can define a stochastic process for $dj_t(X)$, and $dj_t(X_t)^k, k\geq 2$:
\begin{table}
\centering
\begin{tabular}{p{1cm}|p{1cm}|p{1cm}|p{1cm}|p{1cm}}
-&$dA^{\dagger}_t$&$d\Lambda_t$&$dA_t$&$dt$\\
\hline
$dA^{\dagger}_t$&0&0&0&0\\
$d\Lambda_t$&$dA^{\dagger}_t$&$d\Lambda_t$&0&0\\
$dA_t$&$dt$&$dA_t$&0&0\\
$dt$&0&0&0&0\\
\end{tabular}
\caption{Ito multiplication operators for the basic operators of quantum stochastic calculus.}\label{ito_table}
\end{table}
\begin{align}\label{dX}
dj_t(X)&=j_t(\alpha^{\dagger})dA^{\dagger}_t+j_t(\alpha)dA_t+j_t(\lambda)d\Lambda_t+j_t(\theta)dt\\
k\geq 2:dj_t(X)^k &=j_t(\lambda^{k-1}\alpha^{\dagger}) dA^{\dagger}_t+j_t(\alpha\lambda^{k-1}) dA_t+j_t(\lambda^k) d\Lambda_t+j_t(\alpha\lambda^{k-2}\alpha^{\dagger}) dt\nonumber\\
\theta &=i[H,X]-\frac{1}{2}\Big(L^*LX+XL^*L-2L^*XL\Big)\nonumber\\
\alpha &=[L^*,X]S\nonumber\\
\alpha^{\dagger} &=S^*[X,L]\nonumber\\
\lambda &=S^*XS-X\nonumber
\end{align}
In order to proceed we first set the system Hamiltonian $H=0$, so that the time evolution of the operator: $X$ arises only from the random noise introduced into the symmetric Fock space. If we then set:
\begin{align}\label{L_op}
L&=\begin{pmatrix} -i\sigma\partial_x&0\\0&-i\sigma\partial_x\end{pmatrix}\\
S&=\mathbb{I}\nonumber
\end{align}
Then we end up with a Gaussian process for $j_t(X)$:
\begin{align}\label{dX_Gaussian}
dj_t(X) &=\begin{pmatrix}0&i\sigma\\-i\sigma&0\end{pmatrix}dA_t+\begin{pmatrix}0&i\sigma\\-i\sigma&0\end{pmatrix}dA^{\dagger}_t\\
dj_t(X)^2 &=\Bigg(\begin{pmatrix}\sigma^2&0\\0&\sigma^2\end{pmatrix}\Bigg)dt\nonumber\\
dj_t(X)^k &=0\text{, }k\geq 3\nonumber
\end{align}
By setting instead:
\begin{align}\label{ext_S}
S(\theta) &=\begin{pmatrix} \cos(\theta)&-\sin(\theta)\\\sin(\theta)&\cos(\theta)\end{pmatrix}\text{, with }\theta=\pi/2
\end{align}
We get:
\begin{align*}
\lambda&=S^*XS-X\\
&=\begin{pmatrix} x+\cos(\pi)\epsilon/2 & -\sin(\pi)\epsilon/2\\-\sin(\pi)\epsilon/2 & x-\cos(\pi)\epsilon/2\end{pmatrix}-\begin{pmatrix} x+\epsilon/2 & 0\\0&x-\epsilon/2\end{pmatrix}\\
&=\begin{pmatrix} x-\epsilon/2&0\\0&x+\epsilon/2\end{pmatrix}-\begin{pmatrix} x+\epsilon/2 & 0\\0&x-\epsilon/2\end{pmatrix}\\
&=\begin{pmatrix} -\epsilon&0\\0&\epsilon\end{pmatrix}
\end{align*}
Which in turn leads to:
\begin{align}\label{dX_S}
dj_t(X) &=\begin{pmatrix}0&i\sigma\\-i\sigma&0\end{pmatrix}dA_t+\begin{pmatrix}0&i\sigma\\-i\sigma&0\end{pmatrix}dA^{\dagger}_t+j_t\begin{pmatrix}-\epsilon & 0\\0&\epsilon\end{pmatrix}d\Lambda_t\\
dj_t(X)^k &=\Bigg(\begin{pmatrix}\sigma^2&0\\0&\sigma^2\end{pmatrix}j_t\begin{pmatrix}-\epsilon & 0\\0&\epsilon\end{pmatrix}^{k-2}\Bigg)dt+\Bigg(\begin{pmatrix}0&i\sigma\\-i\sigma&0\end{pmatrix}j_t\begin{pmatrix}-\epsilon & 0\\0&\epsilon\end{pmatrix}^{k-1}\Bigg)dA_t\nonumber\\
&+\Bigg(j_t\begin{pmatrix}-\epsilon & 0\\0&\epsilon\end{pmatrix}^{k-1}\begin{pmatrix}0&i\sigma\\-i\sigma&0\end{pmatrix}\Bigg)dA^{\dagger}_t+\Bigg(j_t\begin{pmatrix}-\epsilon & 0\\0&\epsilon\end{pmatrix}\Bigg)^kd\Lambda_t\nonumber
\end{align}
We model the derivative price as an operator valued function of $j_t(X)$: $V(j_t(X),t)$, and expand as a power series. Then applying the Ito relations from table \ref{ito_table}, with \ref{dX_S}, we get:
\begin{align*}
dV &=\bigg(\frac{\partial V}{\partial t}+\sum_{k\geq 2}\frac{1}{k!}\frac{\partial^k V}{\partial x^k}\begin{pmatrix}\sigma^2&0\\0&\sigma^2\end{pmatrix}\begin{pmatrix}-\epsilon & 0\\0&\epsilon\end{pmatrix}^{k-2}\bigg)dt\\
&+\bigg(\sum_{k\geq 1}\frac{\partial^k V}{\partial x^k}\begin{pmatrix}0&i\sigma\\-i\sigma&0\end{pmatrix}\begin{pmatrix}-\epsilon & 0\\0&\epsilon\end{pmatrix}^{k-1}\bigg)dA_t\nonumber\\
&-\bigg(\sum_{k\geq 1}\frac{\partial^k V}{\partial x^k}\begin{pmatrix}-\epsilon & 0\\0&\epsilon\end{pmatrix}^{k-1}\begin{pmatrix}0&i\sigma&0\\-i\sigma&0\end{pmatrix}\bigg)dA^{\dagger}_t\nonumber\\
&+\bigg(\sum_{k\geq 1}\frac{\partial^k V}{\partial x^k}\begin{pmatrix}0&i\sigma\\-i\sigma&0\end{pmatrix}\begin{pmatrix}-\epsilon & 0\\0&\epsilon\end{pmatrix}^{k}\begin{pmatrix}0&i\sigma\\-i\sigma&0\end{pmatrix}\bigg)d\Lambda_t\nonumber
\end{align*}
Taking expectations over the tensor product of the system quantum state vector: $\psi$, and the symmetric Fock space vector: $\varepsilon$, before equating to zero, we find that:
\begin{align*}
E^{(\psi\otimes\varepsilon)}\bigg[\frac{\partial V}{\partial t}+\sum_{k\geq 2}\frac{1}{k!}\frac{\partial^k V}{\partial x^k}\begin{pmatrix}\sigma^2&0\\0&\sigma^2\end{pmatrix}\begin{pmatrix}(-\epsilon)^{k-2}&0\\0&\epsilon^{k-2}\end{pmatrix}\bigg]&=0
\end{align*}
Setting:
\begin{align*}
|\psi\rangle&=\begin{pmatrix}\psi_0\\ \psi_b\end{pmatrix}
\end{align*}
We get:
\begin{align}\label{NG_QBS}
\frac{\partial V}{\partial t}&+\frac{\sigma^2}{2}\frac{\partial^2 V}{\partial x^2}+\sigma^2\sum_{k\geq 2}\frac{\epsilon^{(2k-2)}}{(2k)!}\frac{\partial^{2k}V}{\partial x^{2k}}\\
&+\sigma^2(||\psi_b||^2-||\psi_o||^2)\sum_{k\geq 2}\frac{(-\epsilon)^{(2k-3)}}{(2k-1)!}\frac{\partial^{(2k-1)}V}{\partial x^{(2k-1)}}=0\nonumber
\end{align}
In \cite{Hicks6}, it is shown that the Fokker-Planck equation associated to the Quantum Black-Scholes equation: \ref{NG_QBS}, is given by:
\begin{align}\label{NG_FP2}
\frac{\partial p}{\partial t}&=\sigma^2\sum_{k\geq 1}\frac{\epsilon^{(2k-2)}}{(2k)!}\frac{\partial^{2k} p}{\partial x^{2k}}+\sigma^2\eta\sum_{k\geq 2}\frac{(-\epsilon)^{(2k-3)}}{(2k-1)!}\frac{\partial^{(2k-1)} p}{\partial x^{(2k-1)}}\\
\eta&=\big(||\psi_o||^2-||\psi_b||^2\big)\nonumber
\end{align}
\section{Power Series Solution}\label{Sln}
To find a solution to \ref{NG_FP2}, we use a trial function:
\begin{align}\label{trial}
p(x,t)&=\frac{a_{00}}{\sqrt{t}}+\sum_{n\geq 1}\sum_{m=2}^{2n}\frac{a_{nm}}{\sqrt{t}}\bigg(\frac{x^{m}}{t^n}\bigg)
\end{align}
We substitute \ref{trial} into \ref{NG_FP2}, and attempt to match the right \& left hand side, thereby generating a sequence relation for the coefficients: $a_{nm}$.
\begin{proposition}\label{prop_case1}
Subject to convergence of the infinite series, equation \ref{trial} is a solution to \ref{NG_FP2}, if the coefficients $a_{nk}$ are given by:
\begin{align*}
a_{12}&=-\frac{a_{00}}{2\sigma^2}\\
\bigg(\frac{1}{2}-n\bigg)a_{(n-1)m}&=\sigma^2\sum_{l=1}^{\lfloor\frac{2n-m}{2}\rfloor}\binom{m+2l}{2l}\epsilon^{2l-2}a_{n(m+2l)}\\
&-\sigma^2\eta\sum_{l=2}^{\lfloor\frac{2n+1-m}{2}\rfloor}\binom{m+2l-1}{2l-1}\epsilon^{2l-3}a_{n(m+2l-1)}
\end{align*}
\end{proposition}
\begin{proof}
Inserting \ref{trial} into the left hand side of \ref{NG_FP2}, gives:
\begin{align*}
\frac{\partial p}{\partial t} &=\sum_{n\geq 0}\bigg(-\frac{1}{2}-n\bigg)\frac{1}{\sqrt{t}}\sum_{m=0}^{2n}a_{nm}\bigg(\frac{x^m}{t^{n+1}}\bigg)
\end{align*}
Similarly, inserting \ref{trial} into the right hand side of \ref{NG_FP2}, gives:
\begin{align*}
\sigma^2\sum_{k\geq 1}\frac{\epsilon^{(2k-2)}}{(2k)!}\frac{\partial^{2k} p}{\partial x^{2k}}&=\sum_{n\geq 0}\frac{\sigma^2}{\sqrt{t}}\sum_{m=2}^{2n}\sum_{l=1}^{\lfloor m/2\rfloor}\binom{m}{2l}\epsilon^{(2l-2)}a_{nm}\bigg(\frac{x^{(m-2l)}}{t^n}\bigg)\\
\sigma^2\eta\sum_{k\geq 2}\frac{(-\epsilon)^{(2k-3)}}{(2k-1)!}\frac{\partial^{(2k-1)} p}{\partial x^{(2k-1)}}&=\sum_{n\geq 0}\frac{\sigma^2\eta}{\sqrt{t}}\sum_{m=2}^{2n}\sum_{l=1}^{m}\binom{m}{2l-1}\epsilon^{(2l-3)}a_{nm}\bigg(\frac{x^{(m-2l+1)}}{t^n}\bigg)
\end{align*}
Combining the two, we get:
\begin{align}\label{core_eqn}
\sum_{n\geq 0}\big(-\frac{1}{2}-n\big)\frac{1}{\sqrt{t}}\sum_{m=0}^{n}a_{nm}\Big(\frac{x^{m}}{t^{n+1}}\Big)&=\sum_{n\geq 0}\frac{\sigma^2}{\sqrt{t}}\sum_{m=2}^{2n}\sum_{l=1}^{\lfloor m/2\rfloor}\binom{m}{2l}\epsilon^{(2l-2)}a_{nm}\Big(\frac{x^{(m-2l)}}{t^n}\Big)\nonumber\\
+\sum_{n\geq 0}\frac{\sigma^2\eta}{\sqrt{t}}&\sum_{m=2}^{2n}\sum_{l=2}^{\lfloor (m+1)/2\rfloor}\binom{m}{2l-1}\epsilon^{(2l-3)}a_{nk}\bigg(\frac{x^{(m-2l+1)}}{t^n}\bigg)
\end{align}
In order to derive a series to calculate the coefficients $a_{nm}$ we compare coefficients of:$\frac{x^m}{t^n}$ on each side of \ref{core_eqn}. From the left hand side we have:
\begin{align*}
\bigg(\frac{1}{2}-n\bigg)\frac{a_{(n-1)m}}{\sqrt{t}}\bigg(\frac{x^m}{t^n}\bigg)
\end{align*}
Similarly, from the right hand side we have:
\begin{align*}
&\frac{\sigma^2}{\sqrt{t}}\sum_{l=1}^{\lfloor\frac{2n-m}{2}\rfloor}\binom{m+2l}{2l}\epsilon^{2l-2}a_{n(m+2l)}\bigg(\frac{x^m}{t^n}\bigg)\\
&-\frac{\sigma^2\eta}{\sqrt{t}}\sum_{l=2}^{\lfloor\frac{2n+1-m}{2}\rfloor}\binom{m+2l-1}{2l-1}\epsilon^{2l-3}a_{n(m+2l-1)}\bigg(\frac{x^m}{t^n}\bigg)
\end{align*}
Therefore, equating the coefficients for both sides, we find:
\begin{align}\label{a_nm}
\bigg(\frac{1}{2}-n\bigg)a_{(n-1)m}&=\sigma^2\sum_{l=1}^{\lfloor\frac{2n-m}{2}\rfloor}\binom{m+2l}{2l}\epsilon^{2l-2}a_{n(m+2l)}\\
&-\sigma^2\eta\sum_{l=2}^{\lfloor\frac{2n+1-m}{2}\rfloor}\binom{m+2l-1}{2l-1}\epsilon^{2l-3}a_{n(m+2l-1)}\nonumber
\end{align}
Finally, we can solve for the coefficients: $a_{nk}$ in escalating powers of $t$. For $n=0$ we have:
\begin{align*}
-\frac{a_{00}}{2}=\sigma^2a_{12}
\end{align*}
We assume $a_{11}, a_{10}=0$, and that we know the coefficients for $a_{im}$ for all $m$, for $i\leq (n-1)$, and start with the equation involving: $a_{(n-1)(2n-2)}$. We have $m=2n-2$. Therefore, $\frac{2n-m}{2}=1$, and we have only one term on the right hand side of \ref{a_nm}:
\begin{align*}
\bigg(\frac{1}{2}-n\bigg)a_{(n-1)(2n-2)}&=\sigma^2\binom{2n}{2}a_{n(2n)}
\end{align*}
Therefore, from the value of $a_{(n-1)(2n-2)}$ we can calculate the value of  $a_{n(2n)}$.

Now assume, as well as knowing all the coefficients $a_{ij}$ with $i\leq (n-1)$, we know those with $i=n$ and $j=2n$ down to $j=m+4$. Then in equation \ref{a_nm}, there is only one unknown coefficient: $a_{n(m+2)}$.

Thus by the second induction, we can calculate the rest of the coefficients $a_{nj}$ for all $j$, and by the first induction, we can calculate all coefficients: $a_{ij}$, with $i\geq n$.
\end{proof}
\section{Convergence Properties}\label{Conv}
In order to apply proposition \ref{prop_case1}, we investigate the solution to the truncated partial differential equation. For example, with zero skew (number of buyers \& sellers is balanced) we would have:
\begin{align}\label{trunc}
\frac{\partial p_K}{\partial t}&=\sigma^2\sum_{k=1}^K\frac{\epsilon^{(2k-2)}}{(2k)!}\frac{\partial^{2k} p_K}{\partial x^{2k}}
\end{align}
\begin{proposition}\label{prop_trunc}
The power series $p_K(x,t)$, given by:
\begin{align}\label{trial2}
p_K(x,t)&=\frac{a_{00}}{\sqrt{t}}+\sum_{n\geq 1}\sum_{m=2(n-K+1)}^{2n}\frac{a_{nm}}{\sqrt{t}}\bigg(\frac{x^{m}}{t^n}\bigg)
\end{align}
is a solution to the truncated partial differential equation: \ref{trunc}, where the coefficients are given by:
\begin{align}\label{a_nm_trunc}
\bigg(\frac{1}{2}-n\bigg)a_{(n-1)m}&=\sigma^2\sum_{l=1}^{\min(\lfloor\frac{2n-m}{2}\rfloor,K)}\binom{m+2l}{2l}\epsilon^{2l-2}a_{n(m+2l)}
\end{align}
\end{proposition}
\begin{proof}
Each term on the right hand side of \ref{a_nm}, derives from a partial derivative: $\partial^{2l}/\partial x^{2l}$. Equation \ref{a_nm_trunc}, follows by restricting $l\leq K$.

As described in the proof to proposition \ref{prop_case1}, we can proceed as follows:
\begin{itemize}
\item By setting $n=0,m=0$, we can calculate the value for $a_{12}$. Since $\frac{2n-m}{2}=1$, then this is the only non-zero term for $n=1$.
\item For $n=2$, we first calculate the value for $a_{24}$ by setting $m=2$.
\item If $\epsilon=0$, then equations \ref{a_nm} and \ref{a_nm_trunc} are the same. The only nonzero terms are of the form: $a_{n(2n)}$, and the resulting series is the Taylor expansion for the normal distribution probability density.
\item At each value for $n$, we start by setting, $m=2n-2$. This yields the value for $a_{n(2n)}$. Then proceeding as described, the known value for $a_{(n-1)m}$ determines the value for $a_{n(m+2)}$.
\item The left hand side of \ref{a_nm_trunc} gives $K$ equations: $a_{(n-1)(2n-2)}$, $a_{(n-1)(2n-4)}$, etc down to $a_{(n-1)(2n-4-2K)}$.
\item From these, we determine in turn the non-zero values for $a_{n(2n)}$ down to $a_{n(2n-2K-2)}$, as shown in the proof of proposition \ref{prop_case1}.
\end{itemize}
\end{proof}
For $K=1$, from \ref{prop_trunc}, we get:
\begin{align*}
p_1(x,t)&=\sum_{n\geq 0}\frac{a_{n(2n)}}{\sqrt{t}}\bigg(\frac{x^{2n}}{t^n}\bigg)\\
a_{n(2n)}&=-\frac{a_{(n-1)(2n-2)}}{(2n)\sigma^2}
\end{align*}
which, modulo a normalising constant, is the Taylor series expansion (about $x=0$) for the standard Gaussian probability density.

When one moves from $K=1$ to $K=2$, one includes an additional series:
\begin{align*}
\phi_2(x,t)&=\sum_{n\geq2}\frac{a_{n(2n-2)}}{\sqrt{t}}\bigg(\frac{x^{(2n-2)}}{t^n}\bigg)
\end{align*}
Similarly, when moving from $K=2$ to $K=3$ we add a third term:
\begin{align*}
\phi_3(x,t)&=\sum_{n\geq 3}\frac{a_{n(2n-4)}}{\sqrt{t}}\bigg(\frac{x^{(2n-4)}}{t^n}\bigg)
\end{align*}
Now, consider the power series \ref{trial2} as a function of the variable $y=1/t$:
\begin{align}\label{asymptotic}
p(x,y)&=a_{00}\sqrt{y}+\sum_{j\geq 1}\phi_j(x,y)\\
p_K(x,y)&=a_{00}\sqrt{y}+\sum_{j=1}^{K}\phi_j(x,y)\nonumber\\
\phi_j(x,y)&=\sum_{n\geq j}a_{n(2n-2j+2)}(x^{(2n-2j+2)})y^{n+0.5}\nonumber
\end{align}
\begin{proposition}\label{prop_diverge}
The series from proposition \ref{prop_trunc} is a divergent asymptotic expansion for the solution to equation \ref{NG_FP2}, with $\eta=0$.
\end{proposition}
\begin{remark}\label{its_ok}
In this proposition, we show that the series from proposition \ref{prop_trunc} is an asymptotic expansion in the sense of definition 10.1.1 from \cite{Dettman}. That is we show that in equation \ref{asymptotic}, we have:
\begin{align*}
\phi_j(x,y)&=o(\phi_{j-1}(x,y))\text{, as }y\rightarrow 0
\end{align*}
Thus, for a fixed (and arbitrarily high) value for $K$, the truncation error (from ignoring $\phi_j(x,y)$ for $j\geq K+1$) tends to zero for $y\rightarrow 0$. In other words, the approximation becomes more and more accurate for higher values of $t$.

However, for a fixed value of $x$ and $t$, the series diverges as $K\rightarrow\infty$. In section \ref{NumSln}, we show that in practical applications it will be sufficient to include a small number of terms in approximating the solution. 

Proposition \ref{pade} is then crucial in the sense that this enables us to calculate a cut-off time (dependent on $x$), in order to ensure the approximation error remains below a specified level. The solution should then only be applied for times above this cutoff time.
\end{remark}
\begin{proof}[Proof of Proposition \ref{prop_diverge}]
We have from equation \ref{a_nm}, that:
\begin{align}\label{start}
\Big(\frac{1}{2}-n\Big)a_{(n-1)2}&=\sigma^2\sum_{l=1}^{n-1}\binom{2l+2}{2l}\epsilon^{2l-2}a_{n(2+2l)}
\end{align}
If $\sum_{j\geq 1}\phi_j(x,y)$ is a convergent series, then we must have: $\sum_{j\geq 1}\phi_j(1,1)$ is also a convergent series. Therefore, we have that: $\sum_{j\geq 1}\sum_{k=0}^{2j}a_{jk}$ also converges. Let us write the series by ordering the $a_{jk}$ first by $j$ and then by $k$. We write:
\begin{align*}
S_N=\sum_{n=1}^Nb_n
\end{align*}
Where $b_1=a_{00}$, $b_2=a_{10}$, $b_3=a_{11}$, $b_4=a_{12}$, etc. Since we assume that $S_N$ converges, we must have that $b_n\rightarrow 0$ as $n\rightarrow\infty$. Therefore, we can choose $N$ such that: $|b_m|<|b_n|$, for $n>N$ and $m>n$.

Therefore, we can choose $n>N$, such that:
\begin{align}\label{contra_1}
max_{(k\leq n)}a_{n(2k)}&=a_{n,max}\\
&<a_{(n-1)2}\nonumber
\end{align}
Now we have:
\begin{align*}
\sigma^2\sum_{l=1}^{n-1}\binom{2l+2}{2l}\epsilon^{2l-2}a_{n(2+2l)}&=\frac{\sigma^2}{\epsilon^4}\sum_{l=1}^{n-1}\binom{2l+2}{2l}\epsilon^{2l+2}a_{n(2+2l)}\\
&\leq\frac{\sigma^2}{\epsilon^4}\sum_{l=1}^{n-1}\Big\lvert\binom{2l+2}{2l}\epsilon^{2l+2}a_{n(2+2l)}\Big\rvert\\
&\leq\frac{a_{n,max}\sigma^2}{\epsilon^4}\sum_{l=1}^{n-1}\Big\lvert\binom{2l+2}{2l}\epsilon^{2l+2}\Big\rvert
\end{align*}
Now the series:
\begin{align*}
R_n(\epsilon)&=\sum_{l=1}^{n-1}\Big\lvert\binom{2l+2}{2l}\epsilon^{2l+2}\Big\rvert
\end{align*}
is a convergent series for $|\epsilon|<1$, by the ratio test. Therefore, we have for $n>N+1$:
\begin{align}\label{contra_2}
a_{n,max}&\geq\bigg(\frac{\epsilon^4\Big(n-\frac{1}{2}\Big)}{\sigma^2R_{\infty}(\epsilon)}\bigg)a_{(n-1)2}
\end{align}
However, for large enough $n$, we have that equation \ref{contra_2} contradicts equation \ref{contra_1}. Therefore the series: $\sum_{j\geq 1}\sum_{k=0}^{2j}a_{jk}$ is not convergent.

To show that $p_K(x,y)$ is asymptotic to $p(x,y)$ in equation \ref{asymptotic}, as $y\rightarrow 0$, note that for all $j\geq 1$ we have:
\begin{align*}
\phi_j(x,y)&=O(y^{j+0.5})\text{, as }j\rightarrow 0\\
y^j&=O(\phi_{j-1}(x,y))\text{, as }j\rightarrow 0
\end{align*}
Therefore as $y\rightarrow 0$ we have:
\begin{align*}
\phi_j(x,y)&\leq K_1y^{j+0.5}\text{, for some constant }K_1\\
y^j&\leq K_2\phi_{j-1}(x,y)\text{, for some constant }K_2
\end{align*}
So:
\begin{align*}
\phi_j(x,y)&\leq K_1K_2y^{0.5}\phi_{j-1}(x,y)
\end{align*}
Which in turn implies:
\begin{align*}
\phi_j(x,y)&=o(\phi_{j-1}(x,y))\text{, as }y\rightarrow 0
\end{align*}
\end{proof}
\begin{remark}
Note that, since $\epsilon^4=O(R_{\infty}(\epsilon))$ as $\epsilon\rightarrow 1$, the contradiction given by equation \ref{contra_2} is met at smaller values for $n$ as $\epsilon$ increases, and gets closer to $1$. Thus we expect more rapid divergence as $\epsilon$ gets larger (increases from $\epsilon=0$), and that the series will get closer to the Gaussian solution as $\epsilon\rightarrow 0$.
\end{remark}
We now show that, whilst the series given by equation \ref{asymptotic}:
\begin{align*}
S_K(x,y)&=\sum_{j=1}^K\phi_j(x,y)
\end{align*}
is divergent for large $x$, and $y$ (small $t$), each individual term: $\phi_j(x,y)$ does converge for all $x$, and $y$.
\begin{proposition}\label{prop_converge}
The series defined by:
\begin{align*}
\phi_j(x,y)&=\sum_{n\geq j}a_{n(2n-2j+2)}(x^{(2n-2j+2)})y^{n+0.5}
\end{align*}
converges for all $x$ and $y$.
\end{proposition}
\begin{proof}
We write:
\begin{align*}
b_n^k&=a_{n(2n-k)}
\end{align*}
Note that:
\begin{align*}
b_n^0&=-\frac{1}{(2\sigma^2)^nn!}
\end{align*}
So that it is clear that the sequence: $b_n^0$ converges with $O(e^n/n!)$ as $n\rightarrow\infty$. We now assume that this also applies for $b_n^j$ for all $j\leq(k-1)$.

Now, from equation \ref{a_nm}, we have:
\begin{align*}
b_n^k=\bigg[\Big(\frac{1}{2}-n\Big)b_{n-1}^k-\sigma^2\sum_{l=2}^{k+1}\binom{2n-2k-2+2l}{2l}\epsilon^{2l-2}b_n^{k+1-l}\bigg]\bigg(\sigma^2\binom{2n-2}{2}\bigg)^{-1}
\end{align*}
In the summation, we have $k-1$ individual terms, which by assumption, each converge at least to $O(e^n/n!)$. For the first term, we have:
\begin{align*}
\frac{\Big(\frac{1}{2}-n\Big)b_{n-1}^k}{\sigma^2\binom{2n-2}{2}}&=-\frac{(2n-1)}{\sigma^2(2n-2)(2n-3)}b_{n-1}^k
\end{align*}
From which it follows that the $b_n^k$ term also converges with $O(e^n/n!)$. Since we have:
\begin{align*}
\sum_{n=1}^{\infty}y^nx^k\frac{e^n}{n!}
\end{align*}
converges for all $x,y$, it follows that the series:
\begin{align*}
\sqrt{y}\sum_{n=1}^{\infty}b_n^kx^ky^n
\end{align*}
converges, and that therefore: $\phi_j(x,y)$ converges in $n$ for all $x,y,j$.
\end{proof}
As mentioned in remark \ref{its_ok}, we now apply propositions \ref{prop_diverge} and \ref{prop_converge}, to show how to calculate upper bounds for $y$, based on the value for $K$, to ensure the series is convergent and the relative error remains small.
\begin{proposition}\label{pade}
For the series defined in proposition \ref{prop_trunc}, we have:
\begin{align}\label{prop_pade}
\frac{\phi_j(x,y)}{\phi_{j-1}(x,y)}&\approx c_1y+c_2x^2y^2+O(y^3)\\
c_1&=a_{j2}/a_{(j-1)2}\nonumber\\
c_2&=\frac{\big(a_{(j+1)4}-(a_{j2}/a_{(j-1)2})a_{j4}\big)}{a_{(j-1)2}}\nonumber
\end{align}
Therefore, to ensure that: $|\phi_j(x,y)|<\varepsilon|\phi_{j-1}(x,y)|$, we must have:
\begin{align}\label{conv_cond}
|c_1y+c_2y^2x^2|&<\varepsilon\\
c_1&=a_{j2}/a_{(j-1)2}\nonumber\\
c_2&=\frac{\big(a_{(j+1)4}-(a_{j2}/a_{(j-1)2})a_{j4}\big)}{a_{(j-1)2}}\nonumber
\end{align}
\end{proposition}
\begin{proof}
We first write out the ratio of subsequent terms in the series:
\begin{align*}
\phi_j(x,y)/\phi_{j-1}(x,y)
\end{align*}
and invert the Pad{\'e} approximation technique outlined in \cite{Bender} section 8.3. We first write:
\begin{align*}
\frac{\phi_j(x,y)}{\phi_{j-1}(x,y)}&=\frac{\sum_{n=j}^{\infty}a_{n(2n-2j+2)}y^{n+0.5}x^{(2n-2j+2)}}{\sum_{n=j-1}^{\infty}a_{n(2n-2j+4}y^{n+0.5}x^{(2n-2j+4)}}\\
&=\frac{a_{j2}y^{j+0.5}x^2+a_{(j+1)4}y^{j+1.5}x^4+\dots}{a_{(j-1)2}y^{j-0.5}x^2+a_{j4}y^{j+0.5}x^4+\dots}
\end{align*}
We first divide through top \& bottom by $y^{j-0.5}$ to get:
\begin{align}\label{phi_ratio}
\frac{\phi_j(x,y)}{\phi_{j-1}(x,y)}&=\frac{\sum_{k=1}^{\infty}A_k(x)y^{k}}{\sum_{l=0}^{\infty}B_l(x)y^{l}}\\
A_k(x)&=a_{(j+k-1)(2k)}x^{2k}\nonumber\\
B_l(x)&=a_{(j+l-1)(2l+2)}x^{2l+2}\nonumber
\end{align}
We now equate the quotient \ref{phi_ratio}, to a power series in $y$:
\begin{align*}
\sum_{i=1}^{\infty}c_i(x)y^i&=\frac{\sum_{k=1}^{\infty}A_k(x)y^{k}}{\sum_{l=0}^{\infty}B_l(x)y^{l}}
\end{align*}
We can calculate the coefficients: $a_i$ by equating powers of $y$. We have:
\begin{align*}
\Big(\sum_{i=1}^{\infty}c_i(x)y^i\Big)\Big(\sum_{l=0}^{\infty}B_l(x)y^{l}\Big)&=\sum_{k=1}^{\infty}A_k(x)y^{k}
\end{align*}
So that:
\begin{align*}
c_1(x)B_0(x)&=A_1(x)\\
c_2(x)B_0(x)+c_1(x)B_1(x)&=A_2(x)
\end{align*}
From this we get:
\begin{align*}
c_1(x)&=a_{j2}/a_{(j-1)2}\\
c_2(x)&=\frac{\big(a_{(j+1)4}-(a_{j2}/a_{(j-1)2})a_{j4}\big)x^4}{a_{(j-1)2}x^2}
\end{align*}
So that for small $y$ we get:
\begin{align*}
\frac{\phi_j(x,y)}{\phi_{j-1}(x,y)}&\approx c_1y+c_2x^2y^2+O(y^3)\\
c_1&=a_{j2}/a_{(j-1)2}\text{, }c_2=\frac{\big(a_{(j+1)4}-(a_{j2}/a_{(j-1)2})a_{j4}\big)}{a_{(j-1)2}}
\end{align*}
\end{proof}
\section{Numerical Simulations}\label{NumSln}
\subsection{First Results with $\eta=0$:}
In this section, we truncate the trial solution power series to a maximum number of terms in $n$, as well as truncating the partial differential equation:
\begin{align}\label{trunc_trial}
p(x,t)&=\frac{a_{00}}{\sqrt{t}}+\sum_{n=1}^{N}\sum_{m=2(n-K+1)}^{2n}\frac{a_{nm}}{\sqrt{t}}\bigg(\frac{x^m}{t^n}\bigg)
\end{align}
Starting, with a value $N=100$, we plot the solutions for $K=1$ to $K=5$ (terms up to and including $\epsilon^{8}$).

First, figure \ref{1day} shows the 1 day solutions ($t=0.004$), with $\sigma=10\%$, $\epsilon=0.005$.
\begin{figure}
\includegraphics[scale=0.55]{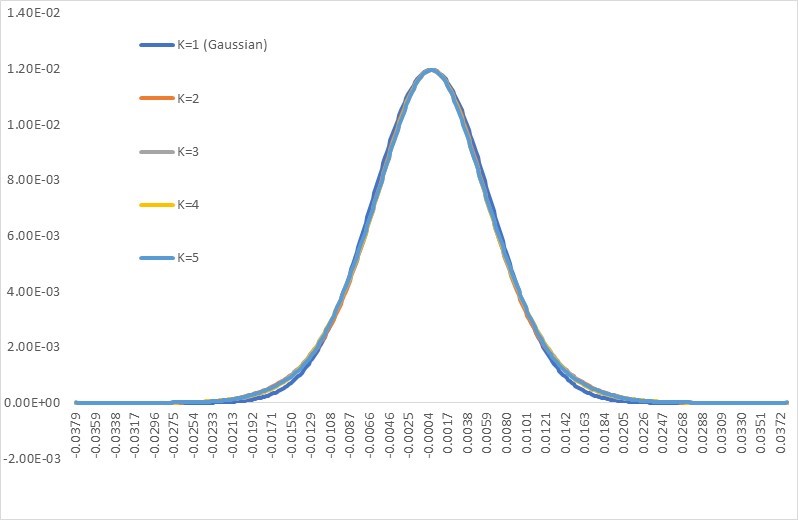}
\includegraphics[scale=0.55]{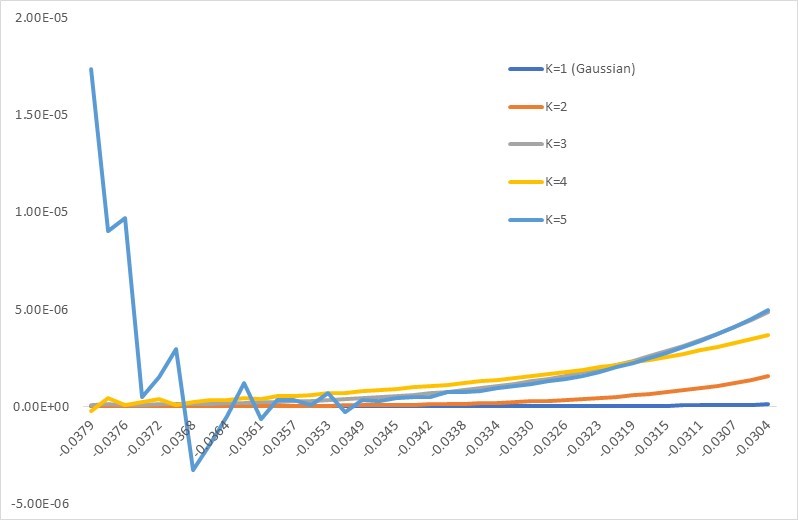}
\caption{Approximate solutions for $N=100$, $K=1$ to $K=5$, t=0.004. The first chart shows the full distribution, the next chart focuses on the left tail.}
\label{1day}
\end{figure}
Next, figure \ref{1m} shows the same solutions for after 1 month has elapsed.
\begin{figure}
\includegraphics[scale=0.55]{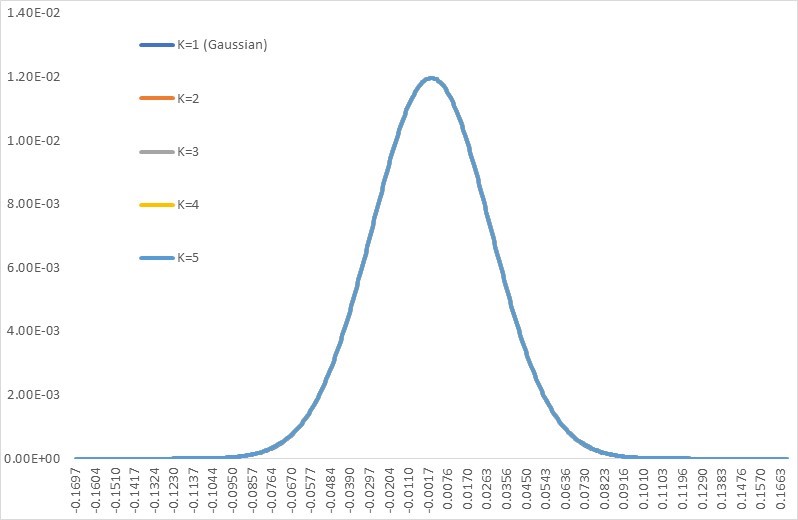}
\includegraphics[scale=0.55]{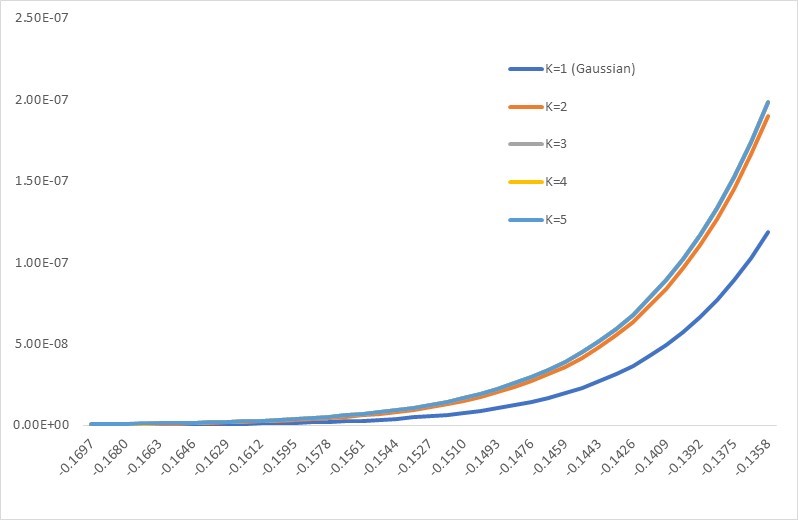}
\caption{Approximate solutions for $N=100$, $K=1$ to $K=5$, t=0.08}
\label{1m}
\end{figure}
\subsection{Convergence in N:}\label{conv in N}
Figure \ref{conv_3} shows the convergence in the tail, for $K=3$. This shows the series has converged for $k\geq 70$.
\begin{figure}
\includegraphics[scale=0.7]{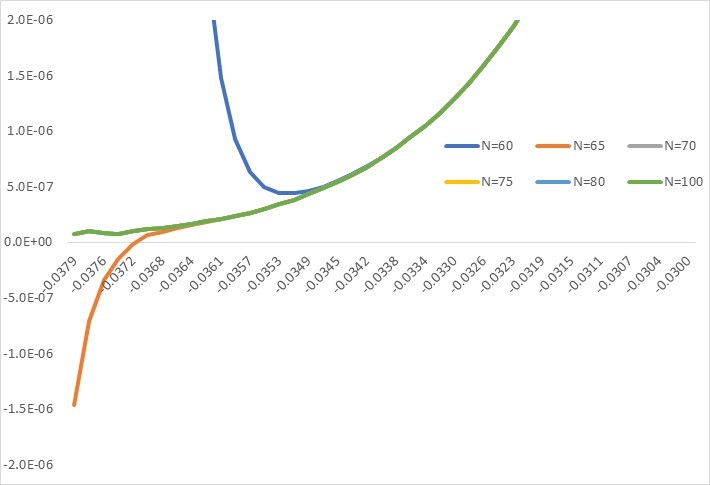}
\caption{Convergence of the tail probabilities, t=0.004, epsilon=0.005, $K=3$}
\label{conv_3}
\end{figure}
Similarly, figure \ref{conv_56} shows the convergence in the tail, for $K=5$.
\begin{figure}
\includegraphics[scale=0.7]{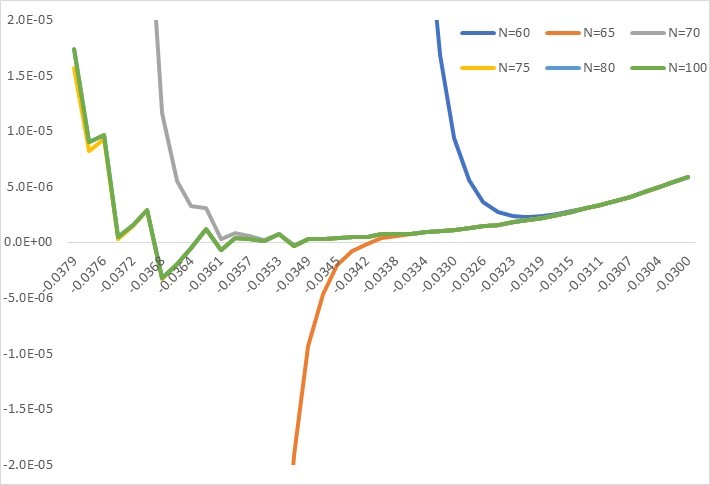}
\caption{Convergence of the tail probabilities, t=0.004, epsilon=0.005, $K=5$}
\label{conv_56}
\end{figure}
We note that, in this case, the power series has converged for $N\geq 75$. However, there is instability in tail for $K=5$, and above. As $K$ increases, the power-series coefficients get larger and larger, the final solution involves subtracting very large numbers from each other.

This is reflected in table \ref{mon_table} below, which shows the maximum value of the contributing monomials, and the ratio of the final sum to the maximum contributing monomial.

The values are taken at 6 standard deviations, and so the final sum of all monomials should be near zero. However, for $K=7$, this involves subtracting monomials with a value of $O\big(10^{+16}\big)$ from each other.

As the size of the individual monomials increases, the number of digits required to capture accuracy to $O\big(10^{-16}\big)$, increases. Thus, eventually the limitations of floating point arithmetic restrict the accuracy of the final result.

\begin{table}
\begin{tabular}{c|c|c}
$k$ ($x=6\sigma$, $t=0.004$)&Max Monomial&Final Sum/Max Monomial\\
\hline
$1$ (Gaussian)&$9.72\mathrm{e}{+7}$&$2.46\mathrm{e}{-15}$\\
$2$&$3.54\mathrm{e}{+9}$&$2.18\mathrm{e}{-15}$\\
$3$&$8.87\mathrm{e}{+10}$&$1.10\mathrm{e}{-15}$\\
$4$&$2.00\mathrm{e}{+12}$&$-1.37\mathrm{e}{-16}$\\
$5$&$4.19\mathrm{e}{+13}$&$5.49\mathrm{e}{-16}$\\
$6$&$8.31\mathrm{e}{+14}$&$3.71\mathrm{e}{-16}$\\
$7$&$1.70\mathrm{e}{+16}$&$2.66\mathrm{e}{-16}$
\end{tabular}
\caption{Table showing the maximum monimial size at value of $x$ within $\pm 6$ std deviations, and the ratio of final sum to the max monomial size.}\label{mon_table}
\end{table}
\subsection{Divergence in K:}\label{conv n K}
The analysis above shows that for fixed $K$, pending sufficient data retention in the floating point arithmetic used, one can use proposition \ref{prop_case1} to calculate solutions.

In this section, we show however that these series diverge for fixed $N$, as $K\rightarrow\infty$. This effect is exacerbated for large $\epsilon$. We show the results in figure \ref{div} below, for the mid-tail probabilities. We set $\epsilon=0.005$, $t=0.004$ (1 day), and $N=100$.
\begin{figure}
\includegraphics[scale=0.64]{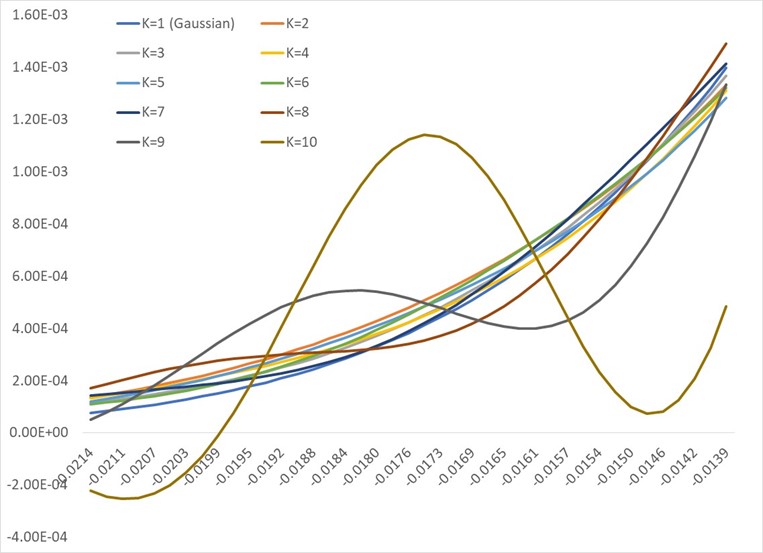}
\caption{Divergence of the mid-tail probabilities, t=0.004, $\epsilon=0.005$}
\label{div}
\end{figure}
Figure \ref{div2} shows the same model after a time frame of 1M has ellapsed. As time increases, the relative of impact of $\epsilon$ versus the total variance: $\sigma^2t$ reduces, and the probability distribution gets closer and closer to the Gaussian distribution. For $t=0.08$, the divergence seen in figure \ref{div} is no longer apparent.
\begin{figure}
\includegraphics[scale=0.6]{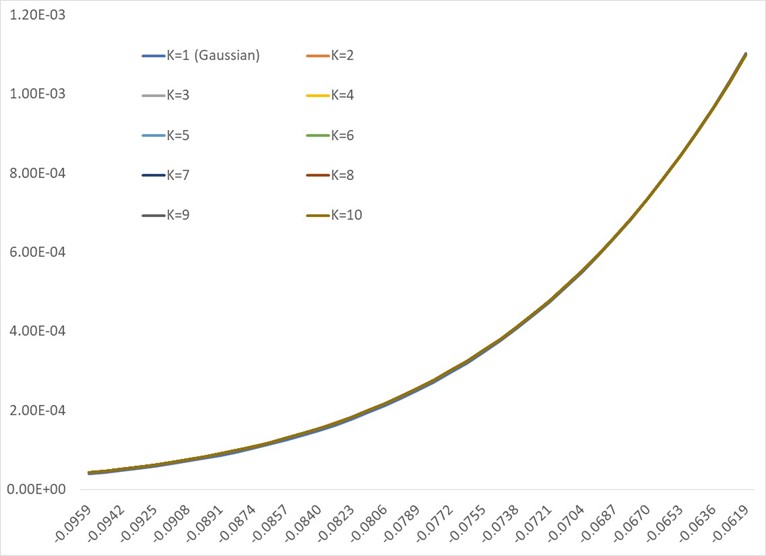}
\caption{Mid-tail probabilities, t=0.08, $\epsilon=0.005$}
\label{div2}
\end{figure}
\subsection{Results with $\eta\neq 0$:}
In figure \ref{skew}, we show the 1 day simulation from figure \ref{1day}: $\sigma=0.1$, $\epsilon=0.005$) with $\eta=0$ to  $\eta=-0.5$.
\begin{figure}
\includegraphics[scale=0.65]{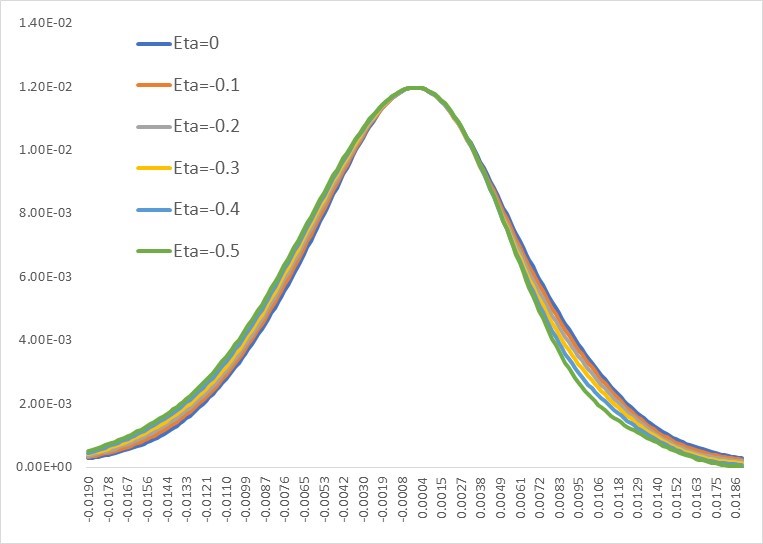}
\caption{Approximate solutions for $N=100$, $K=7$, t=0.004, $\epsilon=0.005$, $\sigma=0.1$, $\eta=0$ to $\eta=-0.5$.}
\label{skew}
\end{figure}
The negative skew parameter of $\eta=-0.5$, reflects the situation whereby the volume of sellers represented by $||\psi_o||^2$, is greater than the volume of buyers. We have:
\begin{align*}
||\psi_o||^2+||\psi_b||^2&=1\\
||\psi_o||^2-||\psi_b||^2&=-0.5
\end{align*}
\section{Application to the Modelling of Illiquid Stocks:}
\subsection{Modelling with `Fat Tails'}
First of all, we note that the solutions converge to the Gaussian distribution for small $\epsilon$, and/or long time frames $t$. In \cite{Hicks6}, it is shown that the second, third \& fourth central moments for the solution to the Fokker-Planck equation: \ref{the pde}, are given by:
\begin{align*}
\mu_2&=\sigma^2t\\
\mu_3&=\sigma^2t\epsilon\eta\\
\mu_4&=3(\sigma^2t)^2+\sigma^2t\epsilon
\end{align*}
Therefore, as $\epsilon/t\rightarrow 0$ the ratio of the kurtosis to the Gaussian kurtosis for a distribution with the same variance, tends to 1:
\begin{align*}
\frac{3(\sigma^2t)^2+\sigma^2t\epsilon}{3(\sigma^2t)^2}\rightarrow 1\text{, as }\frac{\epsilon}{t}\rightarrow 0
\end{align*}
Thus we can see that where the bid-offer spread disappears, and there are a number of buyers \& sellers willing to transact at the same price, the model yields a Gaussian solution. However, after the onset of illiquidity, represented by the fact that market sellers wish to charge a higher price than buyers are willing to pay, the result is higher kurtosis (ie `fat tails').

With this in mind, table \ref{fat_tails} shows the percentiles for $x$ values in excess of $3$, and $4$ standard deviations. The table shows that with $\epsilon=0.005$, the probability of a 1 day move in excess of 4 standard deviations is increased by a factor of $8$. Ie, 1 day every 17 years, rather than 1 day every 134 years.

By contrast, the  probability of a 1 month return in excess of 4 standard deviations is impacted to a much lower degree. In other words, as we look further and further into the future, the current market liquidity, reflected in the width of the bid-offer spread, impacts the likely distribution less.
\begin{table}[H]
\begin{tabular}{c|c|c|c|c|c}
Tail Event&$\epsilon$&$\sigma$&$t$&$K=0$ (Gaussian)&$K=4$\\
\hline
$-3$sd&$0.005$&$0.1$&$0.004$ (1 day)&$0.1374\%$&$0.2758\%$\\
$-4$sd&$0.005$&$0.1$&$0.004$ (1 day)&$0.0030\%$&$0.0240\%$\\
\hline
$-3$sd&$0.005$&$0.1$&$0.08$ (1M)&$0.1417\%$&$0.1577\%$\\
$-4$sd&$0.005$&$0.1$&$0.08$ (1M)&$0.0031\%$&$0.0042\%$
\end{tabular}
\caption{The table shows that with $K=4$, $t=1 day$, the probability of a 3 standard deviation event is 8 times that of the Gaussian distribution.}\label{fat_tails}
\end{table}
\subsection{Model Inaccuracy for Short Time Frames:}
Before applying the solution given by equations \ref{trial}, and proposition \ref{prop_case1}, it must be considered that, whilst this proposition may well represent a solution to the truncated partial differential equation, there is no guarantee that it will not differ substantially from the true solution, or even that it represents a valid probability density function for a stochastic process.

In fact, we can use proposition \ref{pade}, to estimate the minimum time frame for which we can apply the truncated series. For example, if we use a maximum of $K$ terms in the sequence, then in order to ensure $|\phi_{K+1}(x,y)|<\varepsilon|\phi_K(x,y)|$ we at least require:
\begin{align*}
|c_1y+c_2y^2x^2|&<\varepsilon\\
c_1&=a_{(K+1)2}/a_{K2}\\
c_2&=\frac{\big(a_{(K+2)4}-(a_{(K+1)2}/a_{K2})a_{(K+1)4}\big)}{a_{K2}}
\end{align*}
Thus by fixing the error tolerance ($\varepsilon$), we can calculate the maximum value of $y$ (where $y=1/t$) for which the model can be applied.

Alternatively, given a set time frame over which we wish to model, we can use proposition \ref{pade} to calculate how many terms it is safe to include.

With this in mind table \ref{error_table} shows the values of $c_1$ and $c_2$ under different values for $K$ and $\epsilon$.
\begin{table}[H]
\begin{tabular}{c|c|c|c|c}
$K$&$\epsilon$&$c_1$&$c_2$&Minimum $t$, $\varepsilon=5\%$\\
\hline
$1$&$0.005$&$0.0006$&$-0.0365$&$0.0125$\\
$2$&$0.005$&$0.0017$&$-0.1163$&$0.0333$\\
$3$&$0.005$&$0.0028$&$-0.2100$&$0.0562$\\
$4$&$0.005$&$0.0040$&$-0.3086$&$0.08$\\
$5$&$0.005$&$0.0052$&$-0.4093$&$0.1042$\\
\hline
$1$&$0.002$&$0.0001$&$-0.0058$&$0.0020$\\
$2$&$0.002$&$0.0003$&$-0.0186$&$0.0053$\\
$3$&$0.002$&$0.0005$&$-0.0336$&$0.0090$\\
$4$&$0.002$&$0.0006$&$-0.0494$&$0.0128$\\
$5$&$0.002$&$0.0010$&$-0.0655$&$0.0167$
\end{tabular}
\caption{The table shows an estimate for the minimum time over which we can model, assuming a relative error tolerance for higher order terms of 5\%}\label{error_table}
\end{table}
For small $x$ (ie we take $x\approx 0$), $\epsilon=0.005$, and $\varepsilon=5\%$, we find with $K=4$, the minimum modelling time is $0.08$, which equates to roughly 1 month. With $\epsilon=0.002$, we find with $K=4$, the minimum modelling time is 0.0128, which equates to roughly 3 days.

\bibliographystyle{amsplain}


\end{document}